\documentclass[prodmode,acmec]{ec-acmsmall}
\acmVolume{X}
\acmNumber{X}
\acmArticle{X}
\acmYear{2013}
\acmMonth{2}
\usepackage[numbers]{natbib}
\usepackage{subfigure}
\usepackage[ruled, vlined]{algorithm2e}
\usepackage{graphicx,amssymb,amsmath}
\newcommand{\commentout}[1]{}
\newcommand{\eat}[1]{}
\newcommand{\topic}[1]{\vspace{5pt}\noindent{{\bf #1 :}}}

\newcommand{\dealnum}{M} 
\newcommand{\slotnum}{K} 
\newcommand{\deal}{\mathbb{C}} 
\newcommand{\profit}{p} 
\newcommand{\target}{\mathbf{N}}
\newcommand{\N}{\mathbf{N}} 
\newcommand{\alloc}{\mathbf{x}} 

\newcommand{\Nsum}{|\mathbf{N}|} 

\newcommand{\calA}{{\mathcal A}}
\newcommand{\calB}{{\mathcal B}}
\newcommand{\bfx}{{\mathbf x}}
\newcommand{\bfy}{{\mathbf y}}

\newcommand{\poly}{\mathrm{poly}}
\newcommand{\OPT}{\mathcal{OPT}}

\renewcommand{\P}{\mathbf{P}}
\newcommand{\NP}{\mathbf{NP}}

\newcommand{\GAP}{\textsf{GAP}} 

\begin{document}

\title{Optimal Groupon Allocations}
\author{
Weihao Kong \affil{Shanghai Jiao Tong University}
Jian Li \affil{IIIS, Tsinghua University}
Tao Qin \affil{Microsoft Research Asia}
Tie-Yan Liu \affil{Microsoft Research Asia}
}

\begin{abstract}
Group-buying websites represented by Groupon.com are very popular in electronic commerce and online shopping nowadays.
They have multiple slots to provide deals with significant discounts to their visitors every day.
The current user traffic allocation mostly relies on human decisions.
We study the problem of automatically allocating the user traffic of a group-buying website to different deals to maximize the
total revenue and refer to it as the Group-buying Allocation Problem (\GAP).
The key challenge of \GAP\ is how to handle the tipping point (lower bound) and the purchase limit ( upper bound) of each deal. We formulate \GAP\ as a knapsack-like problem with variable-sized items and majorization constraints.
Our main results for \GAP\ can be summarized as follows.
(1) We first show that for a special case of \GAP, in which the lower bound equals the upper bound for each deal,
 there is a simple dynamic programming-based algorithm that can find an optimal allocation in pseudo-polynomial time.
(2) The general case of \GAP\ is much more difficult than the special case.
To solve the problem, we first discover several structural properties of the optimal allocation, %
and then design a two-layer dynamic programming-based algorithm leveraging those properties.
This algorithm can find an optimal allocation in pseudo-polynomial time.
(3) We convert the two-layer dynamic programming based algorithm to a fully polynomial time approximation scheme (FPTAS),
using the technique developed in~\cite{ibarra1975fast}, combined with some careful modifications of the dynamic programs.
Besides these results,
we further investigate some natural generalizations of \GAP,
and propose effective algorithms.
\end{abstract}




\begin{bottomstuff}
\end{bottomstuff}
\maketitle

 \section{Introduction}
 \label{sec:introduction}
Nowadays, group-buying websites are very popular in electronic commerce and online shopping.  They provide users with multiple deals from merchants with significant discounts every day. For each deal, a price after discount, a tipping point (lower bound), a purchase limit (upper bound), and a deadline are specified. Only when the number of purchases exceeds the tipping point is the deal on. A deal is closed once its deadline or purchase limit is reached.

\begin{figure*}[ht]
\centering
\subfigure[Featured deal]{
   \includegraphics[scale =0.35] {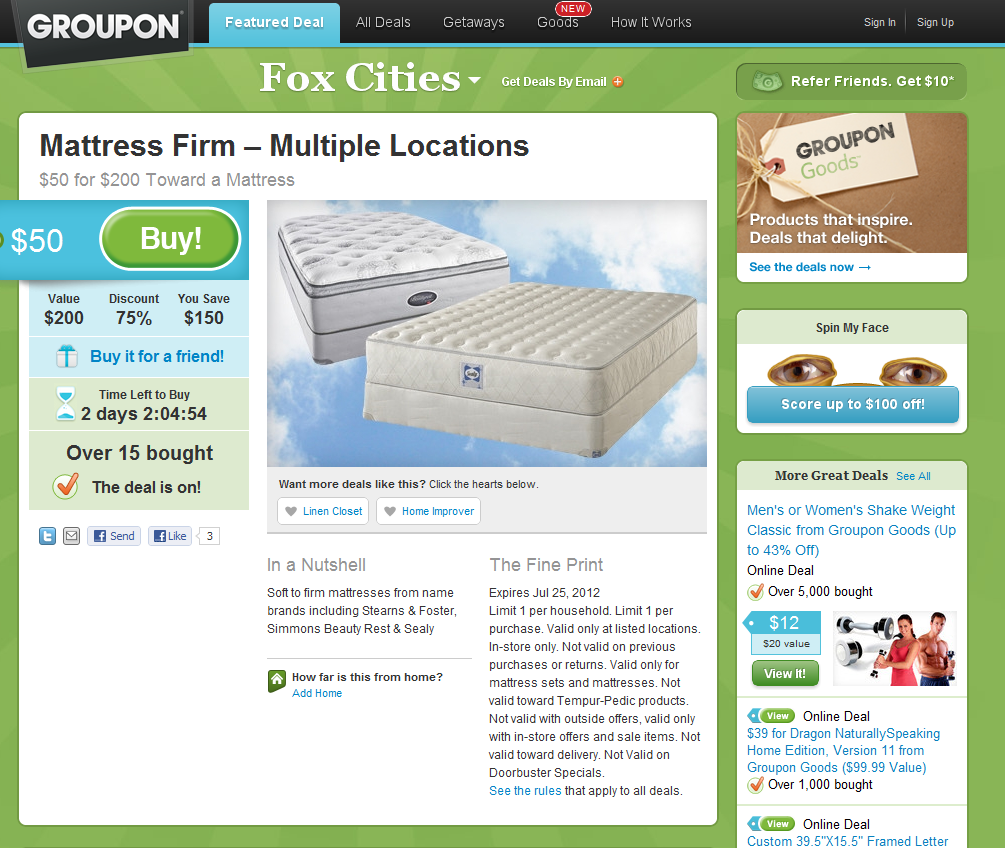}
   \label{fig:feature}
 } \qquad
 \subfigure[Sidebar deals]{
   \includegraphics[scale =0.4] {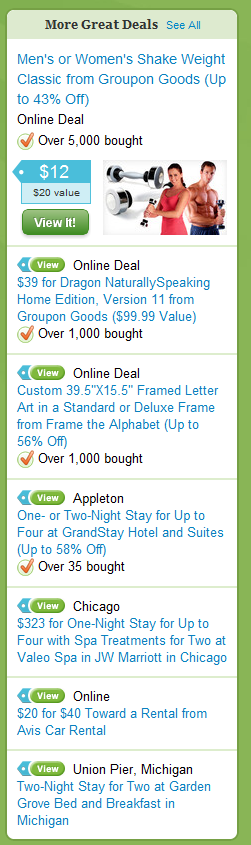}
   \label{fig:sidebar}
 }
\label{fig:sample}
\caption{Example deals from Groupon.com}
\end{figure*}


A group-buying website usually show a featured deal at the main position (see Figure \ref{fig:feature}) of its webpage, and several other deals at the sidebar positions (see Figure \ref{fig:sidebar}). As far as we know, the current practice in the group-buying websites heavily relies on human decisions to lay out those deals: to select one deal as the featured deal and to rank the other deals at sidebar positions. Once the decisions are made, the allocation of deals will be fixed and all the visitors from the same location  (e.g., Boston) will see the same set of deals and the same layout of those deals. We call such a strategy \emph{fixed allocation}. Fixed allocations clearly require a significant amount of human efforts but may still result in suboptimal allocations of user traffic. For example, it is not rare that a visitor comes to the website and finds the featured deal has been sold out. If we can replace this featured deal by another available deal, we can make better use of the website traffic.
That is, if a group-buying website shows different deals to different visitors automatically, it should be able to get better revenue.
We call such a strategy \emph{automatic allocation}.

In this paper, we study how to make automatic allocation of deals for a group-buying website to maximize its revenue. We call such a problem the group-buying Allocation Problem (\GAP), and give it a knapsack-like formulation in Section~\ref{sec:formulation}. \GAP\ is a difficult task because of the following two factors.
\begin{enumerate}
\item We need to display multiple deals, one at each slot, for one visitor, and one deal cannot be shown at more than one slots for one visitor.
The seemingly naive constraints combined with the fact that different slot positions have different conversion biases
directly translate to the majorization constraints (see Section~\ref{sec:constraint}) which are challenging to deal with.
\item Each deal has both a tipping point and a purchase limit. To make money from one deal, the group-buying website must ensure the deal achieves the tipping point.
This corresponds to a lower bound constraint if we decide to allocate some traffic to a particular deal. To maximize revenue, the group-buying website needs to ensure the traffic allocated to a deal will not go beyond its purchase limit. This corresponds to an upper bound constraint.
 \end{enumerate}

Now, we introduce the group-buying allocation problem (\GAP) and give it a knapsack-like formulation.

\subsection{Group-buying Allocation Problem}

\label{sec:formulation}

Suppose that a GS website has $\dealnum$ candidate deals
for a given period (e.g., the coming week) and $N$ website visitors during the time period. Here, for simplicity, we make two assumptions: (1) the number of visitors can be forecasted and is known to the GS website;\footnote{It is an important issue to make accurate forecast of future user traffic, but it is not within the scope of this paper.} and (2) all the deals arrive at the same time and have the same deadline.\footnote{We leave the case of deals with different arriving times and deadlines to the future work.}

For an individual deal $i$, the following information is known to the GS website for allocation.
\begin{itemize}
\item The original price $w_i$ describes the price of the item without the discount associated with the deal.
\item The discount $d_i$ can be specified in terms of a percentage of the original price. A visitor of the website can purchase the deal with price $w_id_i$.
\item The tipping point $L_i$ describes the minimum number of purchases that users are required to make in order for the discount of the deal to be invoked; otherwise, the deal fails and no one can get the deal and the discount.
\item The purchase limit $U_i$ denotes the maximal service capability of a merchant, and at most $u_i$ visitors can purchase this deal. For example, a restaurant can serve at most 200 customers during the lunch time.
\item $s_i$ is the percentage of revenue share that the GS website can get from each transaction. That is, for each purchase of the deal, the revenue of the GS website is $w_id_is_i$. For example, Groupon.com usually gets 30\%-50\% from each transaction.
Note that the GS website can get the revenue from a deal $i$ only if it is on (i.e., at least $L_i$ purchases are achieved).
\item $\lambda_i$ is the conversion probability for the $i$-th deal (i.e., the likelihood that a website visitor will purchase this deal). We assume that $\lambda_i$ is known to the GS website in advance for simplicity.\footnote{It is actually a separate line of research how to predict this conversion probability \cite{wu2009predicting,yuan2011predicting}. We do not consider it in this paper.}
\end{itemize}

Suppose that the GS website shows $K$ deals at $K$ slots to each website visitor. Without loss of generality, we assume that the $i$-th slot is better than the $j$-th slot if $i<j$. For example, the first slot is used to show the featured deal; the second slot corresponds to the first position at the sidebar; the third slot corresponds to the second position at the sidebar, so on and so forth.  We use $\gamma_k$ to denote the conversion bias carried by each position. Similar to the position bias of click probability in search advertising \cite{edelman2005internet, aggarwal2006truthful}, we have that
$$1\geq \gamma_1 > \gamma_2 > ... > \gamma_K\geq 0.$$
We use $N_k$ to indicate the number of {\em effective impressions} for slot position $k$:
$$N_k=N\gamma_k.$$
We have
$N_1> N_2 > \cdots > N_K$.
For simplicity and without much loss of accuracy, we assume that $N_k$ is an integer.
We can think an effective impression as an group-buying deal that actually catches some user's attention.
Note here we assume that the number of purchase of deal $i$ is the effective impressions allocated to deal $i$
times the conversion probability $\lambda_i$.
Let $p_i$ denote the expected revenue that the GS website can get from one effective impression of the $i$-th deal if the deal is on:
\begin{equation}\label{eq:w}
p_i=w_id_is_i\lambda_i.
\end{equation}

We use an integer vector $\alloc$ to denote an allocation, where the $i$-th element
$\alloc_i$ denotes the number of effective impressions allocated to the $i$-th deal.
For any vector $\alloc=\{\alloc_1,\alloc_2,\ldots,\alloc_n\}$, let
$\alloc_{[1]}\geq \alloc_{[2]}\geq \ldots \alloc_{[n]}$ denotes the components
of $\alloc$ in nonincreasing order
(ties are broken in an arbitrary but fixed manner).
Since there are multiple candidate deals and multiple slots,
we need to ensure the feasibility of an allocation. An allocation $\alloc$ is {\em feasible} if it satisfies that:
\begin{enumerate}
\item no more than one deal is assigned to a slot for any visitor, and
\item no deal is assigned to more than one slot for any visitor.
\end{enumerate}
We can easily see that the above constraints are essentially the same in preemptive scheduling of independent tasks on uniform machines. Consider $\dealnum$ jobs with processing requirement $\alloc_i(i=1,\ldots,\dealnum)$ to be processed on $\slotnum$ parallel uniform machines with different speeds $N_j (j =1,\ldots,\slotnum)$. Execution of job $i$ on machine $j$ requires $\alloc_i/N_j$ time units. $\alloc$ is a feasible allocation if and only if the minimum makespan of the preemptive scheduling problem is smaller or equal to $1$. According to \cite{brucker2007scheduling}, the sufficient and necessary conditions for processing all jobs in the interval [0,1] are
\begin{equation}\label{eq:fea1}
\frac{\sum_{j=1}^{\dealnum}\alloc_{[j]}}{ \sum_{j=1}^{\slotnum} N_j}\leq 1,\end{equation}
and
\begin{equation}\label{eq:fea2}
\frac{\sum_{j=1}^i \alloc_{[j]}}{\sum_{j=1}^i N_j}\leq 1 \text{ , for all } i\leq \slotnum .
\end{equation}
Thus, a vector $\alloc$ is a feasible allocation for $\GAP$ if it satisfies the inequalities in Eqn.~(\ref{eq:fea1}) and (\ref{eq:fea2}).

\subsection{Problem Formulation} \label{sec:constraint}

Based on the above notations, finding an optimal allocation means solving the following optimization problem.
\begin{align*}
         &\max_{\alloc}\sum_{i=1}^M p_i \alloc_i\\
s.t.\quad\quad& \frac{L_i}{\lambda_i}\leq \alloc_i\leq \frac{U_i}{\lambda_i} \quad \text{or} \quad \alloc_i=0,  \text{ for } i=1,2,...M\\
         &\alloc \text{ is a feasible allocation.}\\
\end{align*}
Note that $\alloc$ is a vector of integers throughout this paper, and we do not explicitly add it as a constraint when the context is clear. The first set of constraints says the actual number of purchase for deal $i$, which is $\lambda_i \alloc_i$,
should be between the lower bound $L_i$ and upper bound $U_i$. For simplicity, we denote $l_i=\frac{L_i}{\lambda_i}$ and $u_i=\frac{U_i}{\lambda_i}$ in the following sections.

Further, we note that the feasibility conditions in Eqn.~(\ref{eq:fea1}) and (\ref{eq:fea2}) can be exactly described by the majorization constraints.
\begin{definition}\textbf{Majorization constraints}\\
The vector $\bfx$ is {\em majorized}
\footnote{
In fact, the most rigorous term used here should be ``sub-majorize" in
mathematics and theoretical computer science literature
(see e.g.,~\cite{nielsen2002introduction, goel2006simultaneous}).
Without causing any confusion, we omit the prefix for simplicity.
}
by vector $\bfy$
(denoted as $\bfx\preceq \bfy$)
if the sum of the largest $i$ entries in $\bfx$
is no larger than the sum of the largest $i$ entries in $\bfy$
for all $i$, i.e.,
\begin{align}
\label{eq:maj}
\sum_{j=1}^i \bfx_{[j]} \leq \sum_{j=1}^i \bfy_{[j]}.
\end{align}
\end{definition}
In the above definition, $\bfx$ and $\bfy$ should contain the same number of elements. In Eqn.~(\ref{eq:fea1}) and (\ref{eq:fea2}), $N$ has less elements than $\alloc$; one can simply add $\dealnum-\slotnum$ zeros into $N$ (i.e., $N_{[i]}=0, \forall \slotnum<i\leq \dealnum$).

Now we are ready to abstract \GAP\ as an combinatorial optimization problem as the following.
\begin{definition} \textbf{Problem formulation for \GAP}\\
There are $\dealnum$ class of items, $\deal_1,\ldots, \deal_\dealnum$.
Each class $\deal_i$ is associated with a lower bound $l_i\in \mathbb{Z}^+$
and upper bound $u_i\in \mathbb{Z}^+$.
Each item of $\deal_i$ has a profit $\profit_i$.
We are also given a vector $\N=\{N_1,N_2,\ldots, N_\slotnum\}$,
called the {\em target vector}, where $N_1> N_2 > \cdots > N_K$.
We use $\Nsum$ to denote $\sum_{i=1}^{\slotnum}N_j$.
Our goal is to choose $\alloc_i$ items from class $\deal_i$ for each $i\in [\dealnum]$
such that the following three properties hold:
\begin{enumerate}
\item
Either $\alloc_i=0$ (we do not choose any item of class $\deal_i$ at all)
or $l_i\leq \alloc_i\leq u_i$ (the number of items of class $\deal_i$ must satisfy
both the lower and upper bounds);
\item
The vector $\alloc=\{\alloc_i\}_{i}$ is majorized by the target vector $\N$ (i.e., $\alloc\preceq \N$);
\item
The total profit of chosen items is maximized.
\end{enumerate}
\end{definition}

\topic{Our results}
Our major results for \GAP\ can be summarized as follows.
\begin{itemize}
\item[1.] (Section~\ref{sec:simplecase})
As a warmup, we start with a special case of the \GAP\ problem: the lower bound of each class of items equals the upper bound.
In this case, we can order the classes by decreasing lower bounds and the order
enables us to design a nature dynamic programming-based algorithm
which can find an optimal allocation in pseudo-polynomial running time.
\item[2.] (Section~\ref{sec:general})
We then consider the general case of the \GAP\ problem where the lower bound can be smaller than the upper bound.
The general case is considerably more difficult than the simple case in that there is no natural order
to process the classes. Hence, it is not clear how to extend the previous dynamic program to the general case.
To handle this difficulty,  we discover several useful structural properties of the optimal allocation.
In particular, we can show that the optimal allocation can be decomposed into multiple {\em blocks},
each of them has at most one {\em fractional} class (the number of allocated items for the class is less than the upper bound
and larger than the lower bound). Moreover, in a block, we can determine for each class except the fractional class,
whether the allocated number should be the upper bound or the lower bound.
Hence, within each block, we can reduce the problem to the
simpler case where the lower bound of every item equals the upper bound (with slight modifications).
We still need a higher level dynamic program to assemble the blocks and need to show that
no two different blocks use items from the same class.
 Our two level dynamic programming-based algorithm can find an optimal allocation in pseudo-polynomial time.
\item[3.] (Section~\ref{sec:fptas})
Using the technique developed in~\cite{ibarra1975fast}, combined with some careful modifications,
we can further convert the pseudo-polynomial time dynamic program to a fully polynomial time approximation scheme (FPTAS).
We say there is an FPTAS for the problem, if for any fixed constant $\epsilon>0$,
we can find a solution with profit at least $(1-\epsilon)\OPT$ in
$\poly(\dealnum, \slotnum, \log \Nsum, 1/\epsilon )$ time (See e.g., \cite{vazirani2004approximation}).
\item[4.] (Appendix~\ref{app:mono}) We consider the generalization
 from the strict decreasing target vector (i.e., $N_1> N_2 > \cdots > N_K$) to the non-increasing target vector (i.e., $N_1\geq N_2\geq \ldots \geq N_\slotnum$), and briefly describe a pseudo-polynomial time
 dynamic programming-based algorithm for this setting based on the algorithm in Section \ref{sec:general}.
\item[5.] (Appendix~\ref{app:approximation})
For theoretical completeness, we consider
for a generalization of \GAP\
where the target vector $\N=\{N_1,\ldots, N_\slotnum\}$ may be non-monotone.
We provide a $\frac{1}{2}-\epsilon$ factor approximation algorithm for any constant $\epsilon>0$.
In this algorithm, we use somewhat different techniques to handle the majorization constraints,
which may be useful in other variants of \GAP.
\end{itemize}

\section{Related Work}
\label{sec:relatedwork}

\subsection{Studies about Group-buying}
\cite{byers2012groupon} studies how daily deal sites affecting the reputation of a business using evidences from Yelp reviews. It shows that (1) daily deal sites benefit from significant word-of-mouth effects during sales events, and (2) while the number of reviews increases significantly due to daily-deal promotions, average rating scores from reviewers who mention daily deals are 10\% lower than scores of their peers on average. Further, \cite{byers2012daily} investigates hypotheses such as whether group-buying subscribers are more critical than their peers, whether group-buying users are experimenting with services and merchants outside their
usual sphere, or whether some fraction of group-buying merchants provide significantly worse service to customers using group coupons.

 \cite{dholakia2011makes} tries to answer the the question  whether group-buying deals would be profitable for businesses. It shows that because of their many alluring features and the large volume of subscribers (site visitors), group-buying promotions offer a potentially compelling business model, and suggests that for
longer-term success and sustainability, this industry will likely have to design promotions that
better align the deals offered to end consumers with the benefits accruing to the merchants. \cite{edelman2011groupon} finds that offering vouchers is more profitable for merchants which are patient or relatively unknown, and for merchants with low marginal costs.

 There are several papers studying consumer purchase/repurchase behaviors towards group-buying deals. Using the dataset of Groupon.com and Yelp.com, \cite{zhao2012consumer} finds
that, although price promotions offer opportunities for consumers to try new products with a
relatively low cost, online word-of-mouth (WOM) still has a significant impact on product
sales. \cite{lvovskaya2012online} demonstrates  that
consumers consider price as the main contributing factor for a
repurchase at a  local business, after having redeemed a
discount coupon. \cite{song2012daily} finds evidence that social shopping
features deter inexperienced shoppers from buying deals early on.

Another weakly related work is \cite{grabchak2011adaptive}.  Although titled as ``Groupon style'', it actually studies the allocation problem for display advertising but not group-buying services. It formulates the allocation for display advertising as the multi-armed bandit problem, and proposes several greedy policies which can achieve 1/3- or 1/4-approximation bounds.

It is easy to see that the focuses of existing works are very different from ours. We stand on the position of a group-buying website and focus on the problem of revenue maximization by means of designing smart allocation algorithms.

\subsection{Relation to Scheduling and Knapsack Problems}

\GAP\ bears some similarity with the classic
parallel machine scheduling problems \cite{pinedo2008scheduling}.
The $K$ slots can be viewed as $K$ parallel machines with different speeds
(commonly termed as the {\em uniformly related machines}, see, e.g., \cite{ebenlendr2004optimal, horvath1977level}).
The $M$ deals can be viewed as $M$ jobs. Since a deal can be shown at different slots for different visitors, this
means the scheduling can be {\em preemptive} using the language of scheduling \cite{horvath1977level, lawler1978preemptive,gonzalez1978preemptive}.

One major difference between \GAP\ and the scheduling problems lies in the objective functions.
Most scheduling problems target to minimize some functions related to time given the constraint of finishing all the jobs,
such as makespan minimization \cite{ghirardi2005makespan}, total completion time minimization \cite{kuo2006minimizing},
total weighted completion time minimization \cite{schulz1996scheduling}, and total weighted flow time minimization \cite{azizoglu1999minimization,bansal2003algorithms}.
In contrast, our objective is to maximize the revenue generated from the finished jobs (deals in our problem)
given the constraint of limited time. This is similar to the classic knapsack problem in which we want to maximize the total
profit of the items that can be packed in a knapsack with a known capacity.
In fact, our FPTAS in Section~\ref{sec:fptas} borrows the technique from~\cite{ibarra1975fast} for the knapsack problem.
Our work is also related to the {\em interval scheduling} problem \cite{lipton1994online, epstein2012online}
in which the goal is to schedule a subset of interval (preemptively or non-preemptively)
such that the total profit is maximized. \GAP\ differs from this problem in that the intervals (we can think each deal as an interval)
may have variable sizes.


\section{Warmup: A Special Case}
\label{sec:simplecase}
In this section, we investigate a special case of \GAP, in which $l_i=u_i$ for every class.
In other words, we either select a fixed number ($\alloc_i=l_i$) of items from class $\deal_i$, or nothing from the class.
We present an algorithm that can find the optimal allocation in $\poly(\dealnum, \slotnum, \Nsum)$ time
based on dynamic programming.

For simplicity, we assume that the $M$ classes are indexed by the descending order of $l_i$ in this section. That is, we have
$
l_1\geq l_2\geq l_3\geq\cdots \geq l_M.
$

Let $G(i,j,k)$ denote the maximal profit by selecting at most $i$ items from extactly $k$  of the first $j$ classes, which can be expressed by the following integer optimization problem.
\begin{align}
G(i,j,k) \,\,=\,\,    \max_\alloc  &\sum_{t=1}^j \profit_t \alloc_t \notag\\
\text{subject to}\quad& \alloc_t= l_t \quad \text{or} \quad \alloc_t=0, \quad \text{for }1\leq t\leq j   \label{eq:sim1}\\
         &\sum_{t=1}^{r} \alloc_{[t]}\leq \sum_{t=1}^{r} N_t, \quad\text{for } r=1,2,...,\min{\{j,K\}} \label{eq:sim2}\\
         &\sum_{t=1}^{j} \alloc_{[t]} \leq i \label{eq:sim3}\\
	&\alloc_{[k]}>0, \quad
	\alloc_{[k+1]}=0 \label{eq:sim5}
\end{align}
In the above formulation, $\alloc=\{\alloc_1, \alloc_2, ..., \alloc_j\}$ is a $j$ dimensional allocation vector.
$\alloc_{[t]}$ is the $t$-th largest element of vector $\alloc$.
Eqn. (\ref{eq:sim2}) restates the majorization constraints.
Eqn. (\ref{eq:sim3}) ensures that at most $i$ items are selected and
Eqn. (\ref{eq:sim5}) indicates that exactly $k$ classes of items are selected.
Further, we use $Z(i,j,k)$ to denote the optimal allocation vector of the above problem (a $j$-dimensional vector).

It is easy to see that the optimal profit of the special case is $\max_{1\leq k \leq K}G(\Nsum, M, k)$.
In the following, we present an algorithm to compute the values of $G(i,j,k)$ for all $i,j,k$.

\topic{The Dynamic Program} \label{se:simdp}
Initially, we have the base cases that $G(i,j,k)=0$ if $i,j,k$ all equal zero.
For each $1\leq i\leq \Nsum, 1\leq j \leq M, 1\leq k \leq j$,  the recursion of the dynamic program for $G(i,j,k)$ is as follows.
\begin{align}
\label{eq:simdp}
&G(i,j,k) = \notag\\
& \max \left \{
            \begin{array}{lll}
             G(i,j-1,k), & \text{if }j>0 & \text{ (A)} \\
G(i-1,j,k), & \text{if } i>0  &\text{ (B)}\\
            G(i-l_j,j-1,k-1)+l_j p_j, & \text{if }Z(i-l_j,j-1,k-1)\cup l_j \text{ is feasible} & \text{ (C)}\\
            \end{array}\right.
\end{align}
Note that for the case (C) of the above recursion,
we need to check whether adding the $j$-th class in the optimal allocation vector $Z(i-l_j,j-1,k-1)$ is feasible,
i.e., satisfying the majorization constraints in Eqn. (\ref{eq:sim2}).
The allocation vector $Z(i,j,k)$ can be easily determined from the recursion as follows.
\begin{itemize}
\item[$\bullet$] If the maximum is achieved at case (A), we have $Z(i,j,k)_t=Z(i,j-1,k)_t, \forall 1\leq t\leq j-1,$ and $Z(i,j,k)_j=0.$
\item[$\bullet$] If the maximum is achieved at case (B), we have $Z(i,j,k)_t=Z(i-1,j,k)_t, \forall 1\leq t\leq j.$
\item[$\bullet$] If the maximum is achieved at case (C), we have $Z(i,j,k)_t=Z(i-l_j,j-1,k-1)_t, \forall 1\leq t\leq j-1,$ and $Z(i,j,k)_j=l_j.$
\end{itemize}
According to Eqn. (\ref{eq:simdp}), all $G(i,j,k)$ (and thus $Z(i,j,k)$) can be computed in the time\footnote{One can further decrease the complexity of computing all the $G(i,j,k)$'s to $O(\dealnum\Nsum \min(\dealnum, \slotnum))$  by using another recursion equation. We use the recursion equation as shown in Eqn. (\ref{eq:simdp}) considering its simplicity for presentation and understanding.} of $O(\dealnum^2\Nsum)$.

At the end of this section, we remark that the correctness of the dynamic program crucially relies on the fact that
$u_i=l_i$ for all $\deal_i$ and we can process the classes in descending order of their $l_i$s.
However, in the general case where $u_i\ne l_i$, we do not have such a natural order to process the classes
and the current dynamic program does not work any more.



\newcommand{\opt}{\alloc^\star} 
\renewcommand{\OPT}{\mathcal{OPT}} 
\newcommand{\monoopt}{\vec{\alloc^\star}} 

\section{The Exact Algorithm for \GAP}
\label{sec:general}

In this section, we consider the general case of \GAP\ ($l_i\leq u_i$) and present an algorithm that can find
the optimal allocation in $\poly(\dealnum, \slotnum, \Nsum)$ time
based on dynamic programming.
Even though the recursion of our dynamic program appears to be fairly simple,
its correctness relies on several nontrivial structural properties of the optimal allocation
of \GAP. We first present these properties in Section~\ref{subsec:optimalstructure}. Then we show the dynamic program in Section~\ref{subsec:generalDP} and prove its correctness.

\subsection{The Structure of the Optimal Solution}
\label{subsec:optimalstructure}

Before describing the structure of the optimal solution,
we first define some notations.

For simplicity of description,
we assume all $p_i$s are distinct
\footnote{
This is without loss of generality.
If $p_i=p_j$ for some $i\ne j$, we can break tie by adding an infinitesimal value to $p_i$,
which would not affect the optimality of our algorithm in any way.
}
and the $\dealnum$ classes are indexed in the descending order of $p_i$. That is, we have that
$
p_1>p_2>\cdots>p_\dealnum.
$
Note that the order of classes in this section is different from that in Section \ref{sec:simplecase}.

For any allocation vector $\alloc$, $\alloc_i$ indicates the number of items selected from class $i$, and $\alloc_{[i]}$ indicates the $i$-th largest element in vector $\alloc$. For ease of notions, when we say ``class $\alloc_{[i]}$", we actually
refer to the class corresponding to $\alloc_{[i]}$.
In a similar spirit, we slightly abuse the notation
$\profit_{[i]}$ to denote the per-item profit of the class $\alloc_{[i]}$.
For example, $\profit_{[1]}$ is the per-item profit of the class for which we allocate
the most number of items in $\alloc$ (rather than the largest profit). Note that if $\alloc_{[i]}=\alloc_{[i+1]}$, then we put the class with the larger per-item profit before the one with the smaller per-item profit. In other words, if $\alloc_{[i]}=\alloc_{[i+1]}$, then we have $\profit_{[i]}> \profit_{[i+1]}$.


In an allocation $\alloc$, we call class $\deal_i$ (or $\alloc_i$) {\em addable} (w.r.t. $\alloc$) if $\alloc_i<u_i$.
Similarly, class $\deal_i$ (or $\alloc_i$) is {\em deductible} (w.r.t. $\alloc$) if $\alloc_i>l_i$.
A class $\deal_i$ is {\em fractional} if it is both addable and deductible (i.e., $l_i<\alloc_i<u_i$).

Let $\opt$ be the optimal allocation vector.
We start with a simple yet very useful lemma.
\begin{lemma}
\label{lm:prop1}
If a deductible class $\deal_i$ and an addable class $\deal_j$ satisfy $\opt_i>\opt_j$ in the optimal solution $\opt$,
we must have $p_i>p_j$ (otherwise, we can get a better solution by setting $\opt_i=\opt_i-1$ and $\opt_j=\opt_j+1$).
\end{lemma} The proof of lemma is quite straightforward.

The following definition plays an essential role in this section.

\begin{definition} (Breaking Points and Tight Segments)
Let the set of {\em breaking points} for the optimal allocation $\opt$
be
$$
P=\{t \mid \sum_{i=1}^{t} \opt_{[i]} =\sum_{i=1}^{t} N_i \}=\{t_1< t_2< \ldots< t_{|P|}\}.
$$
To simplify the notations for the boundary cases, we let $t_0=0$ and $t_{|P|+1}=\slotnum$.
We can partition $\opt$ into $|P|+1$ {\em tight segments}, $S_1,\ldots, S_{|P|+1}$,
where $S_i=\{\opt_{[t_{i-1}+1]}, \opt_{[t_{i-1}+2]},\ldots, \opt_{[t_i]}\}$. We call $S_{|P|+1}$ the {\em tail  segment}, and $S_1,\ldots,S_{|P|}$ {\em non-tail tight segments}.\qed
\end{definition}

We have the following useful property about the number of items for each class in a non-tail tight segment.
\begin{lemma}
\label{lm:range}
Given a non-tail tight segment $S_k=\{\opt_{[t_{k-1}+1]}, \opt_{[t_{k-1}+2]},\ldots, \opt_{[t_k]}\}$ which spans $N_{t_{k-1}+1}, \ldots, N_{t_{k}}$.  For each class $\deal_i$ that appears in $S_k$ we must have $N_{t_{k-1}+1} \geq \opt_{i}\geq N_{t_k}$.
\end{lemma}
\begin{proof}
From the definition $\sum_{i=1}^{t_k} \opt_{[i]}=\sum_{i=1}^{t_k} N_i$ and majorization constraint $\sum_{i=1}^{t_k-1} \opt_{[i]} \leq \sum_{i=1}^{t_k-1} N_i$ we know that $\opt_{[t_k]}\geq N_{t_k}$. As $\opt_{[t_k]}$ is the smallest in $S_k$, we proved $\opt_{i} \geq N_{t_k}$. Similarly form $\sum_{i=1}^{t_{k-1}} \opt_{[i]}=\sum_{i=1}^{t_{k-1}} N_i$ and $\sum_{i=1}^{t_{k-1}+1} \opt_{[i]} \leq \sum_{i=1}^{t_{k-1}+1} N_i$ we know that $N_{t_{k-1}+1} \geq \opt_{[t_{k-1}+1]}$. As $\opt_{[t_{k-1}+1]}$ is the biggest in $S_k$, we proved $N_{t_{k-1}+1}\geq \opt_{i} $.
\end{proof}
Note that as we manually set $t_{|B|+1}=\slotnum$, the tail segment actually may not be tight. But we still have $N_{t_{k-1}+1} \geq \opt_{i}$.

Let us observe some simple facts about a tight segment $S_k$.
First, there is at most one fractional class.
Otherwise, we can get a better allocation by selecting one more item from the most profitable fractional class and
removing one item from the least profitable fractional class.
Second, in segment $S_k$, if $\deal_i$ is deductible and $\deal_j$ is addable,
we must have $p_i>p_j$ (or equivalently $i<j$) .
Suppose $\deal_{\alpha(S_k)}$ is the per-item least profitable deductible class in $S_k$
and $\deal_{\beta(S_k)}$ is the per-item most profitable addable class in $S_k$.
From the above discussion, we know $\alpha(S_k)\leq \beta(S_k)$.
If $\alpha(S_k) = \beta(S_k)$, then $\alpha(S_k)$ is the only fractional class in $S_k$.
If there is no deductible class in $S_k$, we let $\alpha(S_k)=1$.
Similarly, if there is no addable class in $S_k$, we let $\beta(S_k)=M$.
Let us summarize the properties of tight segments in the lemma below.

\begin{lemma}
\label{lm:segment}
Consider a particular tight segment $S_k$ of the optimal allocation $\opt$. The following properties hold.
\begin{enumerate}
\item There is at most one fractional class.
\item For each class $\deal_i$ that appears in $S_k$ with $i<\beta(S_k)$,
we must have $\opt_i=u_i$.
\item For each class $\deal_i$ that appears in $S_k$ with $i>\alpha(S_k)$,
we must have $\opt_i=l_i$.
\end{enumerate}
\end{lemma}

Now, we perform the following greedy procedure to produce a coarser partition of $\opt$ into disjoint {\em blocks},
$B_1, B_2,\ldots, B_h$, where each block is the union of several consecutive tight segments.
The purpose of this procedure here is to endow one more nice property to the blocks.
We overload the definition of $\alpha(B_i)$ ($\beta(B_i)$ resp.) to denote the index of
the per-item least (most resp.) profitable deductible (addable resp.) class in $B_i$.
We start with $B_1=\{S_1\}$. So, $\alpha(B_1)=\alpha(S_1)$ and $\beta(B_1)=\beta(S_1)$.
Next we consider $S_2$.
If $[\alpha(B_1), \beta(B_1)]$ intersects with $[\alpha(S_2), \beta(S_2)]$,
we let $B_1\leftarrow B_1\cup S_2$.
Otherwise, we are done with $B_1$ and start to create $B_2$ by letting $B_2=S_2$.
Generally, in the $i$-th step, suppose we are in the process of creating block $B_j$ and proceed to $S_i$.
If $[\alpha(B_j), \beta(B_j)]$ intersects with $[\alpha(S_i), \beta(S_i)]$,
we let $B_j\leftarrow B_j\cup S_i$.
Note that the new $[\alpha(B_j), \beta(B_j)]$ is the intersection of old $[\alpha(B_j), \beta(B_j)]$ and $[\alpha(S_i), \beta(S_i)]$.
Otherwise, we finish creating $B_j$ and let the initial value of $B_{j+1}$ be $S_i$.

We list the useful properties in the following critical lemma.
We can see that Property (2) is new (compared with Lemma~\ref{lm:segment}).
\begin{lemma}
\label{lm:block}
Suppose $B_1,\ldots, B_h$ are the blocks created according to the above procedure from the optimal allocation $\opt$, and
$\alpha(B_i)$ and $\beta(B_i)$ are defined as above.
The following properties hold.
\begin{enumerate}
\item Each block has at most one fractional class.
\item $
\alpha(B_1)\leq \beta(B_1) < \alpha(B_2)\leq \beta(B_2) <\ldots< \alpha(B_h)\leq \beta(B_h).
$
\item For each class $\deal_i$ that appears in any block $B_k$ with $i<\beta(B_k)$,
we must have $\opt_i=u_i$.
\item For each class $\deal_i$ that appears in any block $B_k$ with $i>\alpha(B_k)$,
we must have $\opt_i=l_i$.
\end{enumerate}
\end{lemma}
\begin{proof}
Consider block $B_k=S_i\cup S_{i+1}\cup \ldots\cup S_j$.
It is easy to see from the above procedure that
$[\alpha(B_k), \beta(B_k)]=\bigcap_{t=i}^j [\alpha(S_t), \beta(S_t)]$.
If there are two different fractional class (they must from different $S_t$s),
we have $[\alpha(B_k), \beta(B_k)]=\emptyset$, contradicting the procedure for creating $B_k$.
This proves (1).

Now, we prove (2).
Let us first consider two adjacent tight segments $S_{i-1}$ and $S_{i}$.
By Lemma~\ref{lm:prop1}, we have $p_{\alpha(S_{i-1})}>p_{\beta(S_{i})}$ (or equivalently, $\alpha(S_{i-1})<\beta(S_{i})$).
Suppose we are in the $i$th step when we are creating block $B_j$.
We can see that $\alpha(B_j)> \beta(S_k)$ for all $k\geq i$.
This is because $\alpha(B_j)$ is equal to $\alpha(S_{k'})$ for some $k'<i$.
Consider block $B_{j+1}$. We know it consists of several tight segments $S_k$ with $k\geq i$.
So, $\beta(B_{j+1})$ is equal to $\beta(S_{k})$ for some $k\geq i$, which is less than $\alpha(B_j)$.
Moreover, since intervals $[\alpha(B_j),\beta(B_j)]$ are disjoint, we complete the proof of (2).

Properties (3) and (4) are essentially inherited from Lemma~\ref{lm:segment}. \qed
\end{proof}

\subsection{The Dynamic Program}
\label{subsec:generalDP}

Our algorithm for \GAP\ has two levels, both based on dynamic programming.
In the lower level, we attempt to find the optimal allocation for each block.
Then in the higher level, we assemble multiple blocks together to form a global optimal solution.
Lastly, we prove the optimal allocations for these individual blocks do not use one class of items multiple times, thus
can be assembled together.

 \topic{The Lower Level Dynamic Program}
Let us first describe the lower level dynamic program.
Denote $F(i,j,k), \forall 1\leq i \leq j\leq \slotnum, 1\leq k\leq \dealnum$ as the maximal profit generating from the block $B$ which spans $N_i, N_{i+1},\ldots, N_j$ and $\alpha(B)\leq k \leq \beta(B)$. Note here the block $B$ is not one of the blocks created from the optimal allocation $\opt$, but we still require that it satisfies the properties described in Lemma~\ref{lm:block}.
More formally, $F(i,j,k)$ can be written as an integer program in the following form:
\begin{align}
F(i,j,k) \,\,=\,\,    \max  &\sum_{t=1}^M \profit_t \alloc_t \notag\\
\text{subject to}\quad& \alloc_t= u_t \quad \text{or} \quad \alloc_t=0, \quad \text{for }t<k   \label{eq:1}\\
    \quad\quad& \alloc_t= l_t \quad \text{or} \quad \alloc_t=0, \quad\text{for } t>k  \label{eq:2}\\
    \quad\quad& l_t\leq \alloc_t\leq u_t \quad \text{or} \quad \alloc_t=0, \quad\text{for } t=k \label{eq:3}\\
         &\sum_{t=1}^{r} \alloc_{[t]}\leq \sum_{t=i}^{i+r-1} N_t, \quad\text{for } r=1,2,...j-i \label{eq:4}\\
         &\sum_{t=1}^{j-i+1} \alloc_{[t]} = \sum_{t=i}^{j} N_t \label{eq:5}\\
         &\alloc_{[j-i+2]}=0. \label{eq:6}
\end{align}
Constraints \eqref{eq:1} and \eqref{eq:2} correspond to Properties (3) and (4) in Lemma~\ref{lm:block}. The constraint \eqref{eq:3} says $\deal_k$ may be the only fractional constraint.
The constraints \eqref{eq:4} are the majorization constraints. Constraints \eqref{eq:5} and \eqref{eq:6} say $B$ spans $N_i, \ldots, N_j$ with exactly $j-i+1$ class of items.
If $j=\slotnum$ (i.e., it is the last block), we do not have the last two constraints since we may not have to fill all slots, or with fixed number of classes.

To compute the value of $F(i,j,k)$, we can leverage the dynamic program developed in Section~\ref{se:simdp}. The catch is that for any $x_k\in [l_k, u_k]$, according to Eqn.~ \eqref{eq:1} and \eqref{eq:2},
$\alloc_i$ can only take $0$ or a non-zero value (either $u_i$ or $l_i$).
This is the same as making $u_i=l_i$. Therefore, for a given $x_k\in [l_k, u_k]$, the optimal profit $F(i,j,k)$, denoted as $F_{x_k}(i,j,k)$, can be solved by the dynamic program in Section~\ref{se:simdp}.\footnote{
The only extra constraint is \eqref{eq:5}, which is not hard to ensure at all since the dynamic program in Section~\ref{se:simdp} also keeps track of
the number of slots used so far.} Finally, we have
$$ F (i,j,k)=\max_{x_k=0, l_k, l_k+1, l_k+2, \cdots, u_k} F_{x_k}(i,j,k).$$

 \topic{The Higher Level Dynamic Program}
We use $D(j,k)$ to denote the optimal allocation of the following subproblem:
if $j<\slotnum$, we have to fill up exactly $N_1,N_2,\ldots, N_j$ (i.e., $\sum_{i}\alloc_i =\sum_{i=1}^j N_j$)
and $\alpha(B)\leq k$ where $B$ is the last block of the allocation;
if $j=\slotnum$, we only require $\sum_{i}\alloc_i \leq \sum_{i=1}^j N_j$.
Note that we still have the majorization constraints and want to maximize the profit.
The recursion for computing $D(j,k)$ is as follows:
\begin{equation}
\label{eq:dp}
D(j,k) = \max \Bigl\{ \max_{i < j}\{D(i,k-1)+F(i+1, j, k)\}, D(j, k-1) \Bigr\}.
\end{equation}
We return $D(\slotnum, \dealnum)$ as the final optimal revenue of \GAP.

As we can see from the recursion \eqref{eq:dp},
the final value $D(\slotnum, \dealnum)$ is a sum of several $F$ values, say
$F(1, t_1, k_1), F(t_1+1, t_2, k_2), F(t_2+1, t_3, k_3),\ldots$, where $t_1<t_2<t_3<\ldots$ and $k_1<k_2<k_3<
\ldots$. Each such $F$ value corresponds to an optimal allocation of a block.
Now, we answer the most critical question concerning the correctness of the dynamic program:
whether the optimal allocations of the corresponding blocks together form a global feasible allocation?
More specifically, the question is whether one class can appear in two different blocks?
We answer this question negatively in the next lemma.

\begin{lemma}
\label{lm:nooverlap}
Consider the optimal allocations $\alloc^1$ and $\alloc^2$
corresponding to $F(i_1, j_1, k_1)$ and $F(i_2, j_2, k_2)$ respectively,
where $i_1\leq j_1<i_2\leq j_2$ and $k_1< k_2$.
For any class $\deal_i$, it is impossible that both $\alloc^1_i\ne 0$ and $\alloc^2_i\ne 0$ are true.
\end{lemma}
\begin{proof}
We distinguish a few cases. We will use Lemma~\ref{lm:range} on blocks in the following proof.
\begin{enumerate}
\item $i\leq k_1$. Suppose by contradiction that $\alloc^1_i\ne 0$ and $\alloc^2_i\ne 0$.
We always have $\alloc^1_i \leq u_i$.
Since $i\leq k_1<k_2$, again by Lemma~\ref{lm:block}, we have also $\alloc^2_i=u_i$.
Moreover, from Lemma~\ref{lm:range} we know that $\alloc^1_i \geq N_{j_1}> N_{i_2}\geq \alloc^2_i$. This renders a contradiction.
\item $i\geq k_2$. Suppose by contradiction that $\alloc^1_i\ne 0$ and $\alloc^2_i\ne 0$.
By Lemma~\ref{lm:block}, we know $\alloc^1_i=l_i$ and $\alloc^2_i\geq l_i$.
We also have that $\alloc^1_i>N_{i_2} \geq \alloc^2_i$ due to Lemma~\ref{lm:range}, which gives
a contradiction again.
\item $k_1<i<k_2$. Suppose by contradiction that $\alloc^1_i\ne 0$ and $\alloc^2_i\ne 0$.
By Lemma~\ref{lm:block}, we know $\alloc^1_i=l_i$ and $\alloc^2_i=u_i$. We also have the contradiction by
$\alloc^1_i>\alloc^2_i$.
\end{enumerate}
We have exhausted all cases and hence the proof is complete.\qed
\end{proof}

\begin{theorem}
The dynamic program \eqref{eq:dp} computes the optimal revenue for \GAP\
in time $\poly(\dealnum, \slotnum, \Nsum)$.
\end{theorem}
\begin{proof}
By Lemma~\ref{lm:block}, the optimal allocation $\opt$ can be decomposed into several blocks
$B_1, B_2,\ldots, B_h$ for some $h$.
Suppose $B_k$ spans $N_{i_{k-1}+1}, \ldots, N_{i_{k}}$.
Since the dynamic program computes the optimal value,
we have $F(i_{k-1}+1, i_k, \alpha(B_k))\geq \sum_{i\in B_k} p_i \opt_i$.
Moreover, the higher level dynamic program guarantees that
$$
D(\slotnum,\dealnum)\geq
\sum_k F(i_{k-1}+1, i_k, \alpha(B_k))\geq \sum_k\sum_{i\in B_k} p_i \opt_i =\OPT.
$$
By Lemma~\ref{lm:nooverlap}, our dynamic program returns a feasible allocation.
So, it holds that $D(\slotnum,\dealnum)\leq \OPT$. Hence, we have shown that
$D(\slotnum,\dealnum)= \OPT$.
\qed
\end{proof}

\newcommand{\tF}{\tilde{F}}
\newcommand{\tprofit}{\tilde{Q}}

\section{A Full Polynomial Time Approximation Scheme}
\label{sec:fptas}

The optimal algorithm developed in Section~\ref{sec:general} runs in pseudo-polynomial time
(since it is a polynomial of $\Nsum$).
In this section, we present a full polynomial time approximation scheme (FPTAS) for \GAP.
Recall that we say there is an FPTAS for the problem, if for any fixed constant $\epsilon>0$,
we can find a solution with profit at least $(1-\epsilon)\OPT$ in time polynomial in the input size
(i.e., $O(\dealnum + \slotnum\times \log \Nsum )$)
and $1/\epsilon$.
Note that  assuming $\P\ne \NP$ this is the best possible approximation algorithm we can obtain  since
\GAP\ is $\NP$-hard (it is a significant generalization of the $\NP$-hard knapsack problem).

Our FPTAS is based on the 2-level dynamic programs we developed in Section~\ref{subsec:generalDP}.
We observe that it is only the lower level dynamic program that runs in pseudo-polynomial time.
So our major effort is to convert this dynamic program into an FPTAS.
As before, we still need to guarantee at the end that the allocations for these block can be concatenated together
to form a global feasible allocation.
From now on, we fix the small constant $\epsilon>0$.

\topic{FPTAS for the Lower Level Problem}
We would like to approximate the optimal allocation of a block (i.e., $F(i,j,k)$) in polynomial time.
Now, we fix $i,j$ and $k$.
Since we can only get an approximate solution, some properties are lost and we have to modify
the lower level subproblem (LLS) in the following way.
\begin{enumerate}
   \item[L1.] \eqref{eq:1} is strengthened to be $\alloc_t= \{0, u_t\}$,  for  $t<k$ and $N_i\geq u_t \geq N_j$;
   \item[L2.] \eqref{eq:2} is strengthened to be $\alloc_t= \{0, l_t\}$ for  $t>k$  and $N_i\geq l_t \geq N_j$;
   \item[L3.] \eqref{eq:3} is strengthened to be $\alloc_k=0$ or ($l_k\leq \alloc_k\leq u_k$ and $N_i\geq \alloc_k\geq N_j$);
   \item[L4.] \eqref{eq:4} remain the same and \eqref{eq:5} is relaxed to $\sum_{t=1}^{j-i+1} \alloc_{[i]} \leq \sum_{t=i}^{j} N_t$.
\end{enumerate}
Let $F(i,j,k)$ to be the optimal revenue allocation subject to the new set of constraints.
From the proof of Lemma~\ref{lm:nooverlap}, we can see the modifications of \eqref{eq:1} and \eqref{eq:2}
do not change the problem the optimal allocation also satisfies the new constraints.
\eqref{eq:5} is relaxed since we can not keep track of all possible total sizes in polynomial time.
The optimal solution of the modified problem is no less than $F(i,j,k)$.

We first assume we know the value of $\alloc_k$. We will get rid of this assumption later.
For ease of description, we need some notations.
Let $C$ be the set of classes that may participate in any allocation of LLS (those satisfy L1 and L2).
Let $s_t$ be the number of items used in $\deal_t$ if $\deal_t$ participates the allocation.
In other words, $s_t=u_t$ if $s_t$ satisfies L1,
$s_t=l_t$ if $s_t$ satisfies L2,
and
$s_k=\alloc_k$.

Now, we modify the profit of each class.
Let $\tF$ denote the maximal profit by simply taking one class that satisfies the constraints. It is easy to see that
$$\max_{t\in C}p_ts_t\leq\tF\leq F(i,j,k).$$
For any class $\deal_t$ with $t\ne k$, if $\deal_t$ participates the optimal solution of $F(i,j,k)$,
we know how many items are used (either $u_t$ or $l_t$).
So, we associate the entire class $\deal_t$ with a profit (called {\em modified  class profit})
$$
\tprofit_t =
   \lfloor \frac{2\dealnum \profit_t s_t}{\epsilon \tF} \rfloor \text{ for } t\in C
$$
The modified profit of $\deal_t$ can be seen as a scaled and discretized version of the actual profit of $\deal_t$.
\footnote{
It is critical for us to use the discretized profit as one dimension of the dynamic program
instead of the discretized size.
Otherwise, we may violate the majorization constraints (by a small fraction).
Such idea was also used in the FPTAS for the classic knapsack problem \cite{ibarra1975fast}.
}
It is important to note that $\tprofit_t$ is an integer bounded by $O(\dealnum/\epsilon)$
and the maximum total modified profit we can allocate is bounded by $O(\dealnum^2/\epsilon)$.

Everything is in place to describe the dynamic program.
Let $H(t,r,\tprofit)$ be the minimum total size of any allocation for subproblem $F(i,j,k)$ with the following set of additional constraints:
\begin{enumerate}
\item We can only use classes from $\{\deal_1,\ldots, \deal_t\}$;
\item Exactly $r$ different classes participate in the allocation;
\item The total modified profit of the allocation is $\tprofit$;
\item All constraints of LLS are also satisfied.
\end{enumerate}
Initially, $H(t,r,\tprofit)=0$ for $t,r,\tprofit=0$ and $H(t,r,\tprofit)=\infty$ for others.
The recursion of the dynamic program is as follows:
\begin{align*}
H(t,r,\tprofit) = \min \left \{
            \begin{array}{ll}
            H(t-1,r,\tprofit), & \hbox{If we decide $\alloc_i=0$}; \\
            H(t-1,r-1,\tprofit-\tprofit_t)+s_t, & \hbox{If $\deal_t\in C$ and we use $\deal_t$ and} \\
                                                                  & \hbox{$H(t-1,r-1,\tprofit-\tprofit_t)+s_t\leq N_i+\ldots+N_{i+r}$.}
            \end{array}\right.
\end{align*}
The correctness of the recursion is quite straightforward and we omit its proof.
We return the allocation $\alloc$ corresponding to $H(t,r,\tprofit)$ that has a finite value and $\tprofit$ is the highest.
The running time is bounded by $O((i-j)\times \dealnum\times \frac{\dealnum^2}{\epsilon})$
which is a polynomial.

\begin{lemma}
Suppose $\opt$ is optimal allocation corresponding to $F(i,j,k)$
and we know the value of $\opt_k$.
The profit of the allocation $\alloc$ the above dynamic program is at least $(1-\epsilon)F(i,j,k)$.
\end{lemma}
\begin{proof}
We use $I_t$ and $I^{\star}_t$ as the Boolean variables indicating whether
$\deal_t$ participates in the allocations $\alloc$ and $\opt$ respectively.
Since the dynamic program finds the optimal solution with respect to the modified profit, we have
\begin{align*}
\sum_t I_t \tprofit_t \geq \sum_t I^{\star}_t \tprofit_t.
\end{align*}
By the definition of the modified profit, we can see that
\begin{align*}
\sum_t I_t \Bigl(\frac{2\dealnum\profit_t s_t}{ \epsilon \tF} +1\Bigr)\geq
\sum_t I_t \Bigl\lfloor \frac{2\dealnum\profit_t s_t}{ \epsilon \tF} \Bigr\rfloor
\geq \sum_t I^{\star}_t \Bigl\lfloor \frac{2\dealnum\profit_t s_t}{ \epsilon \tF} \Bigr\rfloor
\geq \sum_t I^{\star}_t \Big(\frac{2\dealnum\profit_t s_t}{ \epsilon \tF} -1\Bigr).
\end{align*}
Simple manipulation gives us that
\begin{align*}
\sum_t I_t \profit_t s_t
\geq \sum_t I^{\star}_t \profit_t s_t - 2\dealnum \frac{\epsilon \tF}{ 2\dealnum} \geq (1-\epsilon) F(i,j,k).
\end{align*}
In the last inequality, we use $\tF\leq F(i,j,k)$.
This completes the proof.
\end{proof}

Lastly, we briefly sketch how we get rid of the assumption that $\alloc_k$ is known
(by losing at most an $\epsilon$ fraction of profit).
Enumerating all possible $\alloc_k$ values is not feasible since there are $O(u_k-l_k)$ possibilities in the worst case.
To achieve a polynomial running time, we only try the following set of possible values for $\alloc_k$:
$$D_k=\{0, l_k, u_k, N_i, N_j\}\cup \{\text{all integers in } [l_k, u_k]\cap [N_i, N_j] \text{ with the form} \lfloor(1+\epsilon)^h\rfloor, h\in \mathbb{Z}^+\}.
$$
Clearly, the size of $D_k$ is $O(\log \Nsum)$.
Moreover, for any possible $\alloc_k$ value, we can see that
there is a number in $D_k$ that is at most $\alloc_k$ and at least $(1-\epsilon)\alloc_k$.
Therefore, for any allocation $\alloc$, there is an allocation $\tilde\alloc$ where $\tilde\alloc_k\in D_k$ and the
profit of $\tilde\alloc$ is at least $1-\epsilon$ times the profit of $\alloc$.

\topic{The Higher Level Problem}
The higher level dynamic program is the same as \eqref{eq:dp}, except that we only use the $(1-\epsilon)$-approximation for $F(i,j,k)$.
The correctness of the overall algorithm follows the same line as before.
Since we have enforced constraints L1, L2 and L3, we can still prove Lemma~\ref{lm:nooverlap}
(all arguments in the proof still carry through).
Because we have a $(1-\epsilon)$-approximation for each $F(i,j,k)$
and $D(j,k)$ is a sum of several such $F$ values,
we also have a $(1-\epsilon)$-approximation for $D(j,k)$.
Moreover, the running time for solving this dynamic program is bounded by a polynomial.
In summary, we have the following theorem.

\begin{theorem}
\label{thm:fptas}
There is an FPTAS for \GAP.
In other words, for any fixed constant $\epsilon>0$,
we can find a feasible allocation with revenue at least $(1-\epsilon)\OPT$ in time $\poly(\dealnum, \slotnum,  \log \Nsum, \frac{1}{\epsilon})$
where $\OPT$ is the optimal revenue.
\end{theorem}

 \section{Conclusions and Future Work}
 \label{sec:con}
We have formulated and studied the group-buying allocation problem: finding an optimal allocation for a group-buying  website to maximize its revenue. We have designed two dynamic programming based algorithms, which can find an optimal allocation in pseudo-polynomial time. An FPTAS has also been derived based on the proposed algorithms,
which can handle instances of larger size.
In fact, our preliminary simulation results show that with a modern PC
the algorithm proposed in Section \ref{sec:general} can handle a reasonable size
instance: a group-buying  website with ten slots, hundreds of deals and millions of website visitors.
The FPTAS can handle even larger instances: with tens of slots, hundreds of deals and hundreds of millions of website visitors.
We are conducting more comprehensive experiments and
the detailed report is deferred to the a future version of this work.

There are many research issues related to group-buying allocation which need further investigations. (1) We have studied the offline allocation problem and assumed that the traffic $N$ of a group-buying website is known in advance and all the deals are available before allocation. It is interesting to study the online allocation problem when the traffic is not known in advance
and both website visitors and deals arrive online one by one.
(2) We have assumed that the position discount $\gamma_i$ of each slot and conversion rate $\lambda_i$ of each deal are given. It is worthwhile to investigate how to maximize revenue with unknown position discount and conversion rate.
(3) We have not considered the strategic behaviors of merchants and consumers/visitors.
It is of great interest to study the allocation problem in the setting of auctions.

\bibliographystyle{acmsmall}
\bibliography{groupon}

\begin{thebibliography}{}

\bibitem[\protect\citeauthoryear{Aggarwal, Goel, and Motwani}{Aggarwal
  et~al\mbox{.}}{2006}]{aggarwal2006truthful}
{\sc Aggarwal, G.}, {\sc Goel, A.}, {\sc and} {\sc Motwani, R.} 2006.
\newblock Truthful auctions for pricing search keywords.
\newblock In {\em Proceedings of the 7th ACM conference on Electronic
  commerce}. ACM, 1--7.

\bibitem[\protect\citeauthoryear{Azizoglu and Kirca}{Azizoglu and
  Kirca}{1999}]{azizoglu1999minimization}
{\sc Azizoglu, M.} {\sc and} {\sc Kirca, O.} 1999.
\newblock On the minimization of total weighted flow time with identical and
  uniform parallel machines.
\newblock {\em European Journal of Operational Research\/}~{\em 113,\/}~1,
  91--100.

\bibitem[\protect\citeauthoryear{Bansal}{Bansal}{2003}]{bansal2003algorithms}
{\sc Bansal, N.} 2003.
\newblock Algorithms for flow time scheduling.
\newblock Ph.D. thesis, Carnegie Mellon University.

\bibitem[\protect\citeauthoryear{Brucker}{Brucker}{2007}]{brucker2007scheduling}
{\sc Brucker, P.} 2007.
\newblock {\em Scheduling algorithms\/} Fifth Ed.
\newblock Springer, 124--125.

\bibitem[\protect\citeauthoryear{Byers, Mitzenmacher, and Zervas}{Byers
  et~al\mbox{.}}{2012a}]{byers2012daily}
{\sc Byers, J.}, {\sc Mitzenmacher, M.}, {\sc and} {\sc Zervas, G.} 2012a.
\newblock Daily deals: Prediction, social diffusion, and reputational
  ramifications.
\newblock In {\em Proceedings of the fifth ACM international conference on Web
  search and data mining}. ACM, 543--552.

\bibitem[\protect\citeauthoryear{Byers, Mitzenmacher, and Zervas}{Byers
  et~al\mbox{.}}{2012b}]{byers2012groupon}
{\sc Byers, J.}, {\sc Mitzenmacher, M.}, {\sc and} {\sc Zervas, G.} 2012b.
\newblock The groupon effect on yelp ratings: a root cause analysis.
\newblock In {\em Proceedings of the 13th ACM Conference on Electronic
  Commerce}. ACM, 248--265.

\bibitem[\protect\citeauthoryear{Dholakia}{Dholakia}{2011}]{dholakia2011makes}
{\sc Dholakia, U.} 2011.
\newblock What makes groupon promotions profitable for businesses?
\newblock {\em Available at SSRN 1790414\/}.

\bibitem[\protect\citeauthoryear{Ebenlendr and Sgall}{Ebenlendr and
  Sgall}{2004}]{ebenlendr2004optimal}
{\sc Ebenlendr, T.} {\sc and} {\sc Sgall, J.} 2004.
\newblock Optimal and online preemptive scheduling on uniformly related
  machines.
\newblock {\em STACS 2004\/}, 199--210.

\bibitem[\protect\citeauthoryear{Edelman, Jaffe, and Kominers}{Edelman
  et~al\mbox{.}}{2011}]{edelman2011groupon}
{\sc Edelman, B.}, {\sc Jaffe, S.}, {\sc and} {\sc Kominers, S.} 2011.
\newblock To groupon or not to groupon: The profitability of deep discounts.
\newblock {\em Harvard Business School NOM Unit Working Paper\/}~11-063.

\bibitem[\protect\citeauthoryear{Edelman, Ostrovsky, and Schwarz}{Edelman
  et~al\mbox{.}}{2005}]{edelman2005internet}
{\sc Edelman, B.}, {\sc Ostrovsky, M.}, {\sc and} {\sc Schwarz, M.} 2005.
\newblock Internet advertising and the generalized second price auction:
  Selling billions of dollars worth of keywords.
\newblock Tech. rep., National Bureau of Economic Research.

\bibitem[\protect\citeauthoryear{Epstein, Jez, Sgall, and van Stee}{Epstein
  et~al\mbox{.}}{}]{epstein2012online}
{\sc Epstein, L.}, {\sc Jez, L.}, {\sc Sgall, J.}, {\sc and} {\sc van Stee, R.}
\newblock Online interval scheduling on uniformly related machines.
\newblock Manuscript.

\bibitem[\protect\citeauthoryear{Ghirardi and Potts}{Ghirardi and
  Potts}{2005}]{ghirardi2005makespan}
{\sc Ghirardi, M.} {\sc and} {\sc Potts, C.} 2005.
\newblock Makespan minimization for scheduling unrelated parallel machines: A
  recovering beam search approach.
\newblock {\em European Journal of Operational Research\/}~{\em 165,\/}~2,
  457--467.

\bibitem[\protect\citeauthoryear{Goel and Meyerson}{Goel and
  Meyerson}{2006}]{goel2006simultaneous}
{\sc Goel, A.} {\sc and} {\sc Meyerson, A.} 2006.
\newblock Simultaneous optimization via approximate majorization for concave
  profits or convex costs.
\newblock {\em Algorithmica\/}~{\em 44,\/}~4, 301--323.

\bibitem[\protect\citeauthoryear{Gonzalez and Sahni}{Gonzalez and
  Sahni}{1978}]{gonzalez1978preemptive}
{\sc Gonzalez, T.} {\sc and} {\sc Sahni, S.} 1978.
\newblock Preemptive scheduling of uniform processor systems.
\newblock {\em Journal of the ACM (JACM)\/}~{\em 25,\/}~1, 92--101.

\bibitem[\protect\citeauthoryear{Grabchak, Bhamidipati, Bhatt, and
  Garg}{Grabchak et~al\mbox{.}}{2011}]{grabchak2011adaptive}
{\sc Grabchak, M.}, {\sc Bhamidipati, N.}, {\sc Bhatt, R.}, {\sc and} {\sc
  Garg, D.} 2011.
\newblock Adaptive policies for selecting groupon style chunked reward ads in a
  stochastic knapsack framework.
\newblock In {\em Proceedings of the 20th international conference on World
  wide web}. ACM, 167--176.

\bibitem[\protect\citeauthoryear{Horvath, Lam, and Sethi}{Horvath
  et~al\mbox{.}}{1977}]{horvath1977level}
{\sc Horvath, E.}, {\sc Lam, S.}, {\sc and} {\sc Sethi, R.} 1977.
\newblock A level algorithm for preemptive scheduling.
\newblock {\em Journal of the ACM (JACM)\/}~{\em 24,\/}~1, 32--43.

\bibitem[\protect\citeauthoryear{Ibarra and Kim}{Ibarra and
  Kim}{1975}]{ibarra1975fast}
{\sc Ibarra, O.} {\sc and} {\sc Kim, C.} 1975.
\newblock Fast approximation algorithms for the knapsack and sum of subset
  problems.
\newblock {\em Journal of the ACM (JACM)\/}~{\em 22,\/}~4, 463--468.

\bibitem[\protect\citeauthoryear{Kuo and Yang}{Kuo and
  Yang}{2006}]{kuo2006minimizing}
{\sc Kuo, W.} {\sc and} {\sc Yang, D.} 2006.
\newblock Minimizing the total completion time in a single-machine scheduling
  problem with a time-dependent learning effect.
\newblock {\em European Journal of Operational Research\/}~{\em 174,\/}~2,
  1184--1190.

\bibitem[\protect\citeauthoryear{Lawler and Labetoulle}{Lawler and
  Labetoulle}{1978}]{lawler1978preemptive}
{\sc Lawler, E.} {\sc and} {\sc Labetoulle, J.} 1978.
\newblock On preemptive scheduling of unrelated parallel processors by linear
  programming.
\newblock {\em Journal of the ACM (JACM)\/}~{\em 25,\/}~4, 612--619.

\bibitem[\protect\citeauthoryear{Lipton}{Lipton}{1994}]{lipton1994online}
{\sc Lipton, R.} 1994.
\newblock Online interval scheduling.
\newblock In {\em Proceedings of the fifth annual ACM-SIAM symposium on
  Discrete algorithms}. Society for Industrial and Applied Mathematics,
  302--311.

\bibitem[\protect\citeauthoryear{Lvovskaya, Tan, and Zhong}{Lvovskaya
  et~al\mbox{.}}{2012}]{lvovskaya2012online}
{\sc Lvovskaya, Y.}, {\sc Tan, S.}, {\sc and} {\sc Zhong, C.} 2012.
\newblock Online discount coupon promotions \& repurchasing behaviors: The
  groupon case.
\newblock Ph.D. thesis, M{\"a}lardalen University.

\bibitem[\protect\citeauthoryear{Nielsen}{Nielsen}{2002}]{nielsen2002introduction}
{\sc Nielsen, M.} 2002.
\newblock An introduction to majorization and its applications to quantum
  mechanics.

\bibitem[\protect\citeauthoryear{Pinedo}{Pinedo}{2008}]{pinedo2008scheduling}
{\sc Pinedo, M.} 2008.
\newblock {\em Scheduling: theory, algorithms, and systems}.
\newblock Springer Verlag.

\bibitem[\protect\citeauthoryear{Schulz}{Schulz}{1996}]{schulz1996scheduling}
{\sc Schulz, A.} 1996.
\newblock Scheduling to minimize total weighted completion time: Performance
  guarantees of lp-based heuristics and lower bounds.
\newblock {\em Integer Programming and Combinatorial Optimization\/}, 301--315.

\bibitem[\protect\citeauthoryear{Song, Park, Yoo, and Jeon}{Song
  et~al\mbox{.}}{2012}]{song2012daily}
{\sc Song, M.}, {\sc Park, E.}, {\sc Yoo, B.}, {\sc and} {\sc Jeon, S.} 2012.
\newblock Is the daily deal social shopping?: An empirical analysis of purchase
  and redemption time of daily-deal coupons.
\newblock {\em Working paper\/}.

\bibitem[\protect\citeauthoryear{Vazirani}{Vazirani}{2004}]{vazirani2004approximation}
{\sc Vazirani, V.} 2004.
\newblock {\em Approximation algorithms}.
\newblock springer.

\bibitem[\protect\citeauthoryear{Wu and Bolivar}{Wu and
  Bolivar}{2009}]{wu2009predicting}
{\sc Wu, X.} {\sc and} {\sc Bolivar, A.} 2009.
\newblock Predicting the conversion probability for items on c2c ecommerce
  sites.
\newblock In {\em Proceeding of the 18th ACM conference on Information and
  knowledge management}. ACM, 1377--1386.

\bibitem[\protect\citeauthoryear{Yuan, Chen, and Mathieson}{Yuan
  et~al\mbox{.}}{2011}]{yuan2011predicting}
{\sc Yuan, T.}, {\sc Chen, Z.}, {\sc and} {\sc Mathieson, M.} 2011.
\newblock Predicting ebay listing conversion.
\newblock In {\em Proceedings of the 34th international ACM SIGIR conference on
  Research and development in Information}. ACM, 1335--1336.

\bibitem[\protect\citeauthoryear{Zhao, Chen, and Ye}{Zhao
  et~al\mbox{.}}{2012}]{zhao2012consumer}
{\sc Zhao, D.}, {\sc Chen, W.}, {\sc and} {\sc Ye, Q.} 2012.
\newblock Consumer purchase behavior under price promotion: Evidence from
  groupon daily deals.
\newblock In {\em SIGBPS Workshop on Business Processes and Services (BPS'12)}.
  141.

\end{thebibliography}

\appendix
\newcommand{\size}{N}
\newcommand{\sol}{\mathcal{SOL}}
\newcommand{\sfM}{{\mathsf M}}
\newcommand{\sfA}{{\mathsf A}}
\newcommand{\sfB}{{\mathsf B}}

\section{Handling General Nonincreasing Target Vector}\label{app:mono}
In this section, we consider the general case
where $N_1\geq N_2\geq \ldots \geq N_\slotnum$ and some inequalities
hold with equality.
Our previous algorithm does not work here since the proof of Lemma~\ref{lm:nooverlap} relies on strict inequalities.
In the general case, the lemma does not hold and we can not
guarantee that no class is used in more than one blocks.
In our experiment, we also found some concrete instances that make the previous algorithm fail.
In fact, this slight generalization introduces a lot of complications, especially in the boundary cases.
We only briefly sketch some essential changes and omit those tedious (but not difficult) details.

We distinguish two types of blocks, E-blocks and F-blocks.
F-blocks are the same as the blocks we defined in Section~\ref{sec:general} with one additional constraint:
Suppose the F-block spans $N_i,\ldots, N_j$. For any class $\deal_k$ in this block, $\alloc_k> N_j$ if $N_j=N_{j+1}$,
and $\alloc_k< N_i$ if $N_i=N_{i-1}$,
To handle the case where there are some consecutive equal $N$ values, we need E-blocks.
An E-block consists of a maximal set of consecutive classes, say $\alloc_{[i]},\ldots, \alloc_{[j]}$, such that
they are of equal size and form the same number of tight segments.
In other words, we have $\sum_{k=1}^{i-1}\alloc_{[k]}=\sum_{k=1}^{i-1}N_k$ and
$\alloc_{[i]}=\alloc_{[i+1]}=\ldots=\alloc_{[j]}=N_i=\ldots=N_j$.
An E-block differ from an F-block in that a E-block may contain more than one fractional classes.
We can still define $\alpha(E_i)$ and $\beta(E_i)$ for an E-block $E_i$.
But this time we may not have  $\alpha(E_i)\leq \beta(E_i)$ due to the presence of multiple fractional classes.

Now, we describe our new dynamic program.
Let $E(i,j,k_1,k_2)$ be the optimal allocation for an E-block that
spans $N_i, \ldots, N_j$ ($N_i=N_{i+1}=\ldots=N_j)$
and $\beta(E_i)\geq k_1$ and $\alpha(E_i) \leq k_2$.
Computing $E(i,j,k_1,k_2)$ can be done by a simple greedy algorithm
that processes the classes in decreasing order of their profits.
We need to redefine the higher level dynamic program.
 $D(i,k,F)$ represents the optimal allocation for the following subproblem :
we have to fill up exactly $N_1,N_2,\ldots, N_i$ (i.e., $\sum_{j}\alloc_j =\sum_{j=1}^i N_j$)
and $\alpha(B)\leq k$ where $B$ is the last block and is an F-block
(if $i=\slotnum$, we allow $\sum_{j}\alloc_j \leq \sum_{j=1}^i N_j$).
$D(i,k,E)$ represents the optimal allocation of the same subproblem
except the last block $B$ is an E-block.

The new recursion is as follows.
We first deal with the case where the last block is an F-block.
\begin{align}
\label{eq:dp1}
D(i,k,F) = \max_{j < i,l < k}\left \{
            \begin{array}{ll}
            D(j,l,F)+F(j+1,i,k) &  \\
            D(j,l,E)+F(j+1,i,k) & \mbox{}
            \end{array}\right.
\end{align}
The other case is where the last block is an E-block.
\begin{align}
D(i,k,E) = \max_{j < i,l \leq k}\left \{
            \begin{array}{ll}
            D(j,l,F) + E(j+1,i,l+1,k) &  \\
            D(j,l,E) + E(j+1,i,l+1,k) & \mbox{if } N_i\neq N_j
            \end{array}\right.
\end{align}
We can show that for two consecutive (F or E) blocks $B_k$ and $B_{k+1}$ and  in the optimal allocation,
we have $\alpha(B_k)< \beta(B_{k+1})$ (even though we may not have $\alpha(B_k)\leq \beta(B_k)$ for E-blocks).
This is why we set the third argument in $E(j+1,i,l+1,k)$ to be $l+1$.
We can also argue that no two blocks would use items from the same class. 
The proof is similar with Lemma~\ref{lm:nooverlap}.
Moreover, we need to be careful about the boundary cases where $\alpha(E)$ and $\beta(E)$ are undefined.
In such case, their default values need to be chosen slightly differently. For clarity, we omit those details.

\section{A $(\frac{1}{2}-\epsilon)$-Approximation When $\N$ is Non-monotone}
\label{app:approximation}

In the section, we provide a $\frac{1}{2}-\epsilon$ factor approximation algorithm
for a generalization of \GAP\ where the target vector $\N=\{N_1,\ldots, N_\slotnum\}$ may not be monotone
($N_1\geq N_2\geq\ldots\geq N_\slotnum$ may not hold).
We still require that
$\sum_{t=1}^{r} \alloc_{[t]}\leq \sum_{t=1}^{r} N_t$ for all $r$.
Although we are not aware of an application scenario that would require the full generality,
the techniques developed here, which are quite different from those in Section~\ref{sec:general},
may be useful in handling other variants of \GAP\ or problems with similar constraints.
So we provide this approximation algorithm for theoretical completeness.

Note that our previous algorithms does not work for this generalization.
In fact, we even conjecture that the generalized problem is strongly NP-hard
(So it is unlike to have a pseudo-polynomial time algorithm).
Next, we present our algorithm assuming $\Nsum$ is polynomially bounded.
At the end, we discuss how to remove this assumption briefly.

We first transform the given instance to a simplified instance.
The new instance enjoys a few extra nice properties which make
it more amenable to the dynamic programming technique.
In the next section, we present a dynamic program for the simplified instance
that runs in time $\poly(\dealnum, \slotnum, \Nsum)$ for any fixed constant $0<\epsilon<\frac{1}{3}$
($\epsilon$ is the error parameter).

\subsection{Some Simplifications of The Instance}

Let $\OPT$ be the optimal value of the original problem.
We make the following simplifications.
We first can assume that
$
u_i/l_i\leq 1/\epsilon$, for every class $\deal_i$.
Otherwise, we can replace $\deal_i$ by a collection of classs
$\deal_{i1}, \deal_{i2}, \ldots, \deal_{ik}$ where
$l_{i1}=l_i, u_{i1}=l_i/\epsilon$,
$l_{i2}=l_i/\epsilon+1, u_{i2}=l_i/\epsilon^2$,
$l_{i3}=l_i/\epsilon^2+1, \ldots$,
$l_{ik}=l_i/\epsilon^{k-1}+1,u_{ik}=u_i$.
In the following lemma, we show that the optimal solution of the new instance is
at least $\OPT$ and can be transformed into a
feasible solution of the original problem without losing too much profit.

\begin{lemma}
The optimal allocation of the new instance is
at least $\OPT$.
Any allocation with cost $\sol$ of the new instance
can be transformed into a
feasible solution of the original instance with cost at least $(\frac{1}{2}-\epsilon)\sol$ in polynomial time.
\end{lemma}
\begin{proof}
The first part follows easily from the fact
that any feasible allocation of the original instance is also feasible in the new instance.
Now, we prove the second part.
An allocation $\alloc$ for the new instance may be infeasible for the original instance if
we use items from multiple classes out of $\deal_{i1},\ldots, \deal_{ik}$.
Assume $h_i=\max\{j \mid \alloc_{ij}> 0\}$ for all $i$.
We can obtain a feasible allocation for the original problem by only
using items from $\deal_{ih_i}$ for all $i$.
The loss of profit can be bounded by
$$
\sum_i \sum_{j=1}^{h_i-1} u_{ij} \leq \frac{1}{1-\epsilon} u_{i(h_i-1)} \leq \frac{1}{1-\epsilon} l_{ih_i} \leq \frac{1}{1-\epsilon} \alloc_{ih_i}.
$$
The profit we can keep is at least $\sum_i \alloc_{ih_i}$. This proves the lemma.
\qed
\end{proof}

For each class $\deal_i$, let $t_i$ be the largest integer of the form
$\lfloor  (1+\epsilon)^k \rfloor$ that is at most $u_i$.
Let the new upper bound be $\max(l_i, t_i)$.
We can easily see that after the modification of $u_i$,
the optimal value of the new instance is at least $(1-\epsilon)\OPT$.
Moreover, we have the following property:
For each positive integer $T$,
let $\sfM(T)=\{\deal_i\mid l_i\leq T\leq u_i\}$.
Because for all $\deal_i\in \sfM_T$, $u_i$ is either equal to $l_i$ (also equal to $T$)
or in the form of $\lfloor  (1+\epsilon)^k \rfloor$, and $u_i/l_i\leq 1/\epsilon$, we have the following property:
\begin{enumerate}
\item[P1.]
All classes in $\sfM_T$ has at most $O(\log_{1+\epsilon}\frac{1}{\epsilon})=O(1)$ different
upper bounds any any fixed constant $\epsilon>0$.
\end{enumerate}

\begin{corollary}
\label{cor:simplification}
Assuming $\Nsum$ is polynomially bounded,
a polynomial time exact algorithm for the simplified instance (satisfying P1) implies a $\frac{1}{2}-\epsilon$ for the original instance
for any constant $\epsilon>0$.
\end{corollary}

\subsection{An Dynamic Program For the Simplified Instance}

Now, we present a dynamic program for the simplified instance.
Our dynamic program runs in time $\poly(\dealnum, \slotnum, \Nsum)$.

We first need to define a few notations.
We use $\target(i, y)$ to denote the vector
$$
\{\size_1,\size_2,\ldots, \size_{i-1},\size_i+y\}
$$
for every $k$ and $x$.
Each $\target(k, x)$ will play the role of the target vector for
some subinstance in the dynamic program.
Note that the original target vector $\target$ can be written as
$\target(n,0)$.
For each positive integer $T$,
we let $\sfM_T=\{\deal_i\mid l_i\leq T\leq u_i\}$
and $\sfM^t_T=\{\deal_i\mid l_i\leq T\leq u_i=t\}$.
Note that $\sfM^t_T$ are not empty for  $O(\log M)$
different $t$ due to P1.
Let $\sfA^t_T=\{\deal_i\mid l_i=T, u_i=t\}$
and $\sfB^t_T=\{\deal_i\mid l_i<T, u_i=t\}$.
Note that
$\sfM_T=\cup_t \sfM^t_T$ and
$\sfM^t_T=\sfA^t_T\cup \sfB^t_T$.
Due to P1, we can see that for a fixed $T$,
there are at most $O(1)$ different $t$ values
such that $\sfA^t_T$ and $\sfB^t_T$ are nonempty.

We define a subinstance $I(T, \{A_t\}_t, \{B_t\}_t, k, x)$ with the following interpretation of the parameters:
\begin{enumerate}
\item $T$: We require $\alloc_i\leq T$ for all $i$ in the subinstance.
\item $\{A_t\}_t, \{B_t\}_t$: Both $\{A_t\}$ and $\{B_t\}$ are collections of subsets of classes.
Each $A_t \subset \sfA^t_T$ ($B_t \subset \sfB^t_T$ resp.) consists of
the least profitable $|A_t|$ classes in $\sfA^t_T$ ($|B_t |$ classes in $\sfB^t_T$ resp.).
We require that among all classes in $\sfA^t_T$ ($\sfB^t_T$ resp.), only
those classes in $A_t$ ($B_t$ resp.) may participate in the solution.
If $\deal_i$ participates in the solution, we must have $l_i\leq \alloc_i\leq u_i$.
Basically, $\{A_t\}, \{B_t\}$ capture the subset of classes in $\sfM_T$ that
may participates in the solution of the subinstance.
Since there are at most $O(1)$ different $t$ such that $\sfA^t_T$ and $\sfB^t_T$ are nonempty,
we have at most $n^{O(1)}$ such different $\{A_t\}$s and $\{B_t\}$s (for a fixed $T$).
\item Each class $\deal_i$ with $u_i<T$ may participate in the solution.
\item $k,y$: $\target(k,y)$ is the target vector for the subinstance.
\end{enumerate}
We use $\OPT(T, \{A_t\}, \{B_t\}_t, k, x)$ to denote the optimal solution for the subinstance $I(T, \{A_t\}, \{B_t\}_t, k, x)$.

Now, we present the recursions for the dynamic program.
In the recursion,
suppose we want to compute the value $D(T, \{A_t\}, \{B_t\}, i,x)$.
Let $\deal_{a(t)}$ be the $A_t$th least profitable class in $\sfA^t_T$
and $\deal_{b(t)}$ be the $B_t$th least profitable class in $\sfB^t_T$.
We abuse the notation $\{A_{t'}-1\}_t$ to denote
the same set as $\{A_t\}_t$ except that the subset $A_{t'}$
is replaced with $A_{t'}\setminus \deal_{a(t)}$
(i.e., the most profitable class in $A_{t'}$ is removed).
The value of $D(T, \{A_t\}, \{B_t\}, \beta, i,y)$ can be computed as follows:
\begin{align*}
\max_{t'} \left\{
   \begin{array}{ll}
     D(T, \{A_{t'}-1\}_t, \{B_t\}_t, i-1, y+N_i-T)+ p_{a(t)} T, & \hbox{if $A_t>1\wedge y+N_i-T\geq 0$;}\text{\quad(A)}\\
     D(T, \{A_t\}_t, \{B_{t'}-1\}_t, i-1, y+N_i-T)+ p_{b(t)} T, & \hbox{if $B_t>1\wedge y+N_i-T\geq 0$;}\text{\quad(B)}\\
     D(T-1, \{A'_t\}_t, \{B'_t\}_t, i,y), & \hbox{see explanation below;} \text{\quad(C)} \\
   \end{array}
 \right.
\end{align*}
(A) captures the case that $\deal_{a(t)}$ participates in the optimal solution of the subinstance
and $\alloc_{a(t)}=T$.
Similarly, (B) captures the decision $\alloc_{b(t)}=T$.
In case (C), 
$\{A'_t\}_t, \{B'_t\}_t$ are obtained from $\{A_t\}_t, \{B_t\}_t$ as follows:
\begin{enumerate}
\item For any $t\geq T$ and $|B_t|>0$,
let $A'_t \subset B_t$ be the set of classes with lower bound $T-1$
and $B'_t$ be $B_t\setminus A'_t$.
\item
We need to include all classes with upper bound $T-1$. That is to let
$A'_{T-1}=\calA_{T-1}^{T-1}$ and
$B'_{T-1}=\calB_{T-1}^{T-1}$.
\end{enumerate}
Note that this construction simply says $\{A'_t\}$ and $\{B'_t\}$
should be consistent with $\{A_t\}$ and $\{B_t\}$.
(C) captures the case that $\alloc_i<T$ for all $i\in \cup_t\sfA^t_T\cup \sfB^t_T$.
This finishes the description of the dynamic program.
We can see the dynamic program runs in time $\poly(\dealnum, \slotnum, \Nsum)$
since there are at most $O(\Nsum^2\dealnum^{O(1)})$ different subinstances
and computing the value of each subinstance takes constant time.

Now, we show why this dynamic program computes the optimal value
for all subinstances defined above.
In fact, by a careful examination of the dynamic program,
we can see that it suffices to show the following two facts
in subinstance $I(T, \{A_t\}, \{B_t\}, k, y)$.
Recall $\deal_{a(t)}$ is the $A_t$th cheapest class in $\sfA^t_T$
and $\deal_{b(t)}$ is the $B_t$th cheapest class in $\sfB^t_T$.
Fix some $t\geq T$, we have that
\begin{enumerate}
\item
Either $\alloc_i=0$ for all $i\in \sfA^t_T$ or  $\alloc_{a(t)}=T$.
\item
Either $\alloc_i<T$ for all $i\in \sfB^t_T$ or  $\alloc_{b(t)}=T$.
\end{enumerate}
For each $i\in \sfA^t_T$, we can have either $\alloc_i=0$ or $\alloc_i=T$
(since both the lower and upper bounds are $T$).
Hence, if some $\alloc_i$ is set to be $T$, it is better to be
the most profitable one in $A_t$, i.e., $\deal_{a(t)}$.
This proves the first fact.
To see the second fact, suppose that $\alloc_{b(t)}<T$ but
$\alloc_j=T$ for some cheaper $\deal_j\in \sfB^t_T(B_t)$  (i.e., $p_j<p_{b(t)}$)
in some optimal solution of the subinstance.
By increasing $\alloc_{b(t)}$ by $1$ and decreasing $\alloc_j$ by $1$,
we obtain another feasible allocation with a strictly higher profit,
contradicting the optimality of the current solution.
Having proved the correctness of the dynamic program,
we summarize our result in this subsection in the following lemma.
\begin{lemma}
\label{lm:dp}
There is a polynomial time algorithm for the problem if the given instance satisfies P1.
\end{lemma}

Combining Corollary~\ref{cor:simplification} and Lemma~\ref{lm:dp},
we obtain the main result in this section.
Our algorithm runs in pseudo-polynomial time.
To make the algorithm runs in truly polynomial time, we can use the technique
developed in Section~\ref{sec:fptas} using the modified profits as one dimension of the dynamic program
instead of using the total size. The profit loss incurred in this step can be bounded by $\epsilon\cdot \OPT$ in the same way.
The details are quite similar to those in Section~\ref{sec:fptas} and we omit them here.

\begin{theorem}
\label{thm:ptas}
For any constant $\epsilon>0$,
there is a factor $\frac{1}{2}-\epsilon$ approximation algorithm for \GAP\ even
when the target vector $\target$ is not monotone.
\end{theorem}

\end{document}